\documentclass[aps,pra,twocolumn,amsmath,amssymb,nofootinbib,superscriptaddress]{revtex4-1}

\newcommand{\bra}[1]{\langle#1|}
\newcommand{\ket}[1]{|#1\rangle}

\newcommand{\BS}{{\sc BosonSampling}}

\usepackage[pdftex]{graphicx}
\usepackage{mathtools}
\usepackage{mathrsfs}
\usepackage{hyperref}
\usepackage{comment}
\hypersetup{colorlinks = true, linkcolor = [rgb]{0.19411,0.51882,0.667058}, urlcolor = [rgb]{0.125490,0.29542,0.1647058}, citecolor = [rgb]{0.75882,0.37411,0.14117}}
\usepackage{amsthm}

\let\latexchi\chi
\makeatletter
\renewcommand\chi{\@ifnextchar_\sub@chi\latexchi}
\newcommand{\sub@chi}[2]{
  \@ifnextchar^{\subsup@chi{#2}}{\latexchi^{}_{#2}}%
}
\newcommand{\subsup@chi}[3]{
  \latexchi_{#1}^{#3}%
}
\makeatother

\let\latexgamma\gamma
\makeatletter
\renewcommand\gamma{\@ifnextchar_\sub@gamma\latexgamma}
\newcommand{\sub@gamma}[2]{
  \@ifnextchar^{\subsup@gamma{#2}}{\latexgamma^{}_{#2}}%
}
\newcommand{\subsup@gamma}[3]{
  \latexgamma_{#1}^{#3}%
}
\makeatother

\begin{document}

\bibliographystyle{apsrev}

\title{A Quantum Optics Argument for the $\textbf{\#P}$-hardness\\ of a Class of Multidimensional Integrals}

\author{Peter P. Rohde}
\email[]{dr.rohde@gmail.com}
\homepage{http://www.peterrohde.org}
\affiliation{Centre for Quantum Computation and Intelligent Systems, Faculty of Engineering \& Information Technology, University of Technology Sydney, NSW 2007, Australia}
\affiliation{Hearne Institute for Theoretical Physics and Department of Physics \& Astronomy, Louisiana State University, Baton Rouge, LA 70803}

\author{Dominic W. Berry}
\affiliation{Department of Physics and Astronomy, Macquarie University, Sydney NSW 2113, Australia}
\affiliation{Hearne Institute for Theoretical Physics and Department of Physics \& Astronomy, Louisiana State University, Baton Rouge, LA 70803}

\author{Keith R. Motes}
\email[]{motesk@gmail.com}
\affiliation{Department of Physics and Astronomy, Macquarie University, Sydney NSW 2113, Australia}

\author{Jonathan P. Dowling}
\affiliation{Hearne Institute for Theoretical Physics and Department of Physics \& Astronomy, Louisiana State University, Baton Rouge, LA 70803}

\date{\today}

\frenchspacing
\begin{abstract}
Matrix permanents arise naturally in the context of linear optical networks fed with nonclassical states of light. In this letter we tie the computational complexity of a class of multi-dimensional integrals to the permanents of large matrices using a simple quantum optics argument. In this way we prove that evaluating integrals in this class is \textbf{\#P}-hard. Our work provides a new approach for using methods from quantum physics to prove statements in computer science.
\end{abstract}

\maketitle

Stemming from the seminal work by Aaronson and Arkhipov \cite{bib:AaronsonArkhipov10}, passive linear optical interferometers fed with quantum states of light, have attracted much interest, as a simple approach to implementing a computationally hard problem \cite{bib:Broome20122012, bib:Crespi3, bib:Tillmann4, bib:Spring2, bib:he16}. This so-called {\BS} problem exploits the fact that amplitudes of the output photon-number configurations are related to matrix permanents \cite{bib:Scheel04}, which are \mbox{\textbf{\#P}-hard} for exact computation in the worst case. The well-known Ryser's algorithm for computing permanents requires \mbox{$O(2^n n)$} runtime \cite{bib:Ryser63}. (For an elementary introduction to {\BS} see Ref. \cite{bib:GardBSintro}.) We show here that by considering the {\BS} problem in quantum-optical phase space, that we are able to express the output amplitudes in terms of multidimensional integrals, rather than in terms of matrix permanents. Because these two formalisms are physically equivalent, our result provides a quantum optical inspired insight into the computational complexity of evaluating this class of integrals. (Previous work has also considered {\BS} with states of light other than Fock states, such as Gaussian input states \cite{bib:olson2014sampling, bib:olson2014boson, bib:catSampling, bib:LundScatter13, bib:SalehLOcomp, bib:rahimi2015efficient}.) The computational complexity of this very simple system is of great interest with recent results showing {\BS}-inspired applications to quantum metrology and quantum chemistry simulations \cite{bib:MORDOR, bib:vibSpec}.

Our work in this letter shows broad applications for utilizing quantum optics tools, in particular, and quantum physics paradigms, in general, to pose and to answer questions about the computational complexity of certain mathematical problems.
 

In this work we use quantum optical characteristic functions to represent the output state of \BS\ as a multi-dimensional integral. This integral formalism directly maps to matrix permanents, and in this way we show that these integrals are also \textbf{\#P}-hard. 
Our construction is a new tool for examining open problems regarding the complexity of \BS-like problems. 

Finally, as an example of our formalism, we show that permutation matrices, whose permanents are easy to compute, are also simple to compute with the integral formalism. 

We begin by reviewing the \BS\ formalism. The relationship between output photon-number amplitudes and matrix permanents is easily obtained in the Heisenberg picture. This can be seen by evolving bosonic creation operators via a linear map,
\begin{align}
\hat{U} \hat{a}_j^\dag \hat{U}^\dag \to \sum_{k=1}^m U_{j,k} \hat{a}_k^\dag,
\end{align}
where there are $m$ modes, $\hat{a}_j^\dag$ is the creation operator for the $j^{th}$ mode, and $U$ is an arbitrary \mbox{$m\times m$} unitary matrix, \mbox{$U\in \mathrm{SU}(m)$}. Such a unitary can always be efficiently constructed using \mbox{$O(m^2)$} optical elements \cite{bib:Reck94, bib:MotesLoop}. By applying this unitary map to an input product state of the form,
\begin{align}
\ket{\psi_\mathrm{in}} &= \ket{1_1,\ldots,1_n,0_{n+1},\ldots,0_m} \nonumber \\
&= \hat{a}_1^\dag\ldots \hat{a}^\dag_n \ket{0}^{\otimes m} \nonumber \\
&= \prod_{j=1}^n \hat{a}_j^\dag \ket{0}^{\otimes m},
\end{align}
we find that the output state is,
\begin{align} \label{eq:state_out}
\ket{\psi_\mathrm{out}} &= \hat{U} \ket{\psi_\mathrm{in}} \nonumber \\
&= \prod_{j=1}^n \sum_{k=1}^m U_{j,k} \hat{a}_k^\dag \ket{0}^{\otimes m} \nonumber \\
&= \sum_T \gamma_T \ket{T_1,\ldots,T_m},
\end{align}
where $T$ denotes a photon-number configuration with $T_j$ photons in the $j^{th}$ mode, and total photon number is conserved, \mbox{$\sum_{j=1}^m T_j = n$}. The amplitudes in this highly entangled superposition are related to matrix permanents as,
\begin{align}
\gamma_T = \frac{\mathrm{perm}(U_T)}{\sqrt{T_1!\ldots T_m!}},
\end{align}
where $U_T$ is an \mbox{$n\times n$} submatrix of $U$ depending on the configuration $T$, obtained by taking rows and columns of $U$ corresponding to input and output photons in the configuration \cite{bib:Scheel04, bib:berry2004post}. The respective measurement probabilities are given by \mbox{$\mathcal{P}(T) = |\gamma_T|^2$}. The permanent is given by,
\begin{align}
\mathrm{perm}(U) = \sum_{\sigma\in S_m} \prod_{j=1}^m U_{j,\sigma_j},
\end{align}
where $\sigma$ are the permutations over $m$ elements. 
Specifically, the permanent arises upon symmetrization of the output state during photodetection because of the exchange symmetry of bosons. The number of terms in this superposition scales exponentially with $n$,
\begin{align}
|T| &= \binom{n+m-1}{n} \nonumber \\
&\gtrsim \frac{1}{\sqrt{2\pi}}2^{2n},
\end{align}
further complicating classical simulation. (Here, Stirling's approximation was used.)

An alternate, yet completely equivalent, formalism for modelling the output photostatistics of such a system is using characteristic functions \cite{bib:GerryKnight05}, which represent the state of the system in phase-space, and from which other representations such as the Wigner function can be calculated. This alternate formalism predicts the same outcomes, but expresses them differently, in terms of multidimensional integrals. Because these two formalisms must be equivalent, this may be used as a basis for characterizing a class of integral equations, which must similarly be \mbox{\textbf{\#P}-hard} to evaluate.

Let us begin with any $m$-mode separable input state of the form,
\begin{align} \label{eq:sep_input}
\hat\rho = \hat\rho_1 \otimes \ldots \otimes \hat\rho_m.
\end{align}
The single-mode characteristic $W$ function \cite{bib:GerryKnight05} is defined as,
\begin{align}
\chi_W(\lambda) = \mathrm{tr}[\hat\rho \cdot \hat{D}(\lambda)],
\end{align}
where \mbox{$\hat{D}(\lambda)$} is the displacement operator, given by,
\begin{align}
\hat{D}(\alpha) = \mathrm{exp}(\lambda \hat{a}^\dag - \lambda^* \hat{a}),
\end{align}
and $\lambda$ is an arbitrary complex number representing the amplitude of the displacement in phase-space. This definition straightforwardly generalizes to the multi-mode case as,
\begin{align}
\chi_W(\lambda_1,\ldots,\lambda_m) = \mathrm{tr}[\hat\rho \cdot \hat{D}_1(\lambda)\ldots \hat{D}_m(\lambda_m)],
\end{align}
where \mbox{$\hat{D}_j(\lambda_j)$} is the displacement operator on the $j^{th}$ mode. Then, the characteristic function, as shown in Appendix \ref{app:disp_ev}, for the state evolved via linear optics is,
\begin{align}
\chi_W^U(\lambda_1,\ldots,\lambda_m) &= \mathrm{tr}[\hat{U} \hat\rho \hat{U}^\dag \cdot \hat{D}_1(\lambda_1)\ldots \hat{D}_m(\lambda_m)] \nonumber \\
&= \mathrm{tr}[\hat\rho \cdot \hat{U}^\dag \hat{D}_1(\lambda_1)\ldots \hat{D}_m(\lambda_m) \hat{U}] \nonumber \\
&= \mathrm{tr}[\hat\rho \cdot \hat{D}_1(\mu_1)\ldots \hat{D}_m(\mu_m)],
\end{align}
where,
\begin{align} \label{eq:beta_from_alpha}
\mu_j = \sum_{k=1}^m \lambda_k U_{j,k}.
\end{align}
When the input state $\hat\rho$ is separable, as per Eq. (\ref{eq:sep_input}), 
with \mbox{$\hat\rho=(\ket{1}\bra{1})^{\otimes n}\otimes (\ket{0}\bra{0})^{\otimes (m-n)}$} so there are $n$ single photons in the first $n$ modes,
the multi-mode characteristic $\chi$ function reduces to,
\begin{align} \label{eq:INTmmCW}
\chi_W^U(\lambda_1,\ldots,\lambda_m) &= \mathrm{tr}[\hat\rho_1\cdot \hat{D}_1(\mu_1)] \ldots \mathrm{tr}[\hat\rho_m\cdot \hat{D}_m(\mu_m)] \nonumber \\
&= \prod_{j=1}^n \bra{1} \hat{D}(\mu_j) \ket{1} \prod_{j=n+1}^m \bra{0}\hat{D}(\mu_j)\ket{0}
\nonumber \\
&= \prod_{j=1}^n e^{-\frac{1}{2}|\mu_j|^2}\left(1-|\mu_j|^2\right) \prod_{j=n+1}^m e^{-\frac{1}{2}|\mu_j|^2} \nonumber \\
&= \prod_{j=1}^m e^{-\frac{1}{2}|\mu_j|^2} \prod_{j=1}^n \left(1-|\mu_j|^2\right) \nonumber \\
&= e^{-\frac{1}{2}\sum_{j=1}^m|\mu_j|^2} \prod_{j=1}^n \left(1-|\mu_j|^2\right) \nonumber \\
&= e^{-\frac{1}{2}\mathcal{E}(\vec\lambda)} \prod_{j=1}^n \left(1-|\mu_j|^2\right),
\end{align}
where we have used the identity of Lemma 2, as shown in Appendix \ref{app:disp_ov}, and $\mathcal{E}$ is the total energy of the system with amplitudes $\vec\lambda$ (or equivalently $\vec\mu$ due to energy conservation),
\begin{align}
\mathcal{E}(\vec\lambda) = \sum_{j=1}^m|\mu_j|^2 = \sum_{j=1}^m|\lambda_j|^2.
\end{align}

Note that the characteristic $W$ function of Eq.~\eqref{eq:INTmmCW}, $\chi_W^U$, can \emph{always} be efficiently calculated with \emph{any} separable input state, since it has a factorized form, and hence there is no exponential growth in the number of terms. The complexity arises when we wish to extract properties of the state, such as determining individual output amplitudes.

Next, we consider the Wigner function, which may be computed as a type of Fourier transform of $\chi_W$ \cite{bib:GerryKnight05},
\begin{align} \label{eq:SingleModeWigner}
W(\alpha) = \frac{1}{\pi^2} \int e^{\lambda^*\alpha - \lambda\alpha^*}  \chi_W(\lambda) d^2\lambda,
\end{align}
in the single-mode case, which again logically generalizes to the multi-mode case as,
\begin{align}
W(\vec\alpha) = \frac{1}{\pi^{2m}} \idotsint e^{\vec{\lambda}^* \cdot \vec\alpha - \vec{\lambda}\cdot \vec\alpha^*} \chi_W(\vec{\lambda}) d^2\vec{\lambda},
\end{align}
where all our complex integrals implicitly run over the range \mbox{$(-\infty,\infty)$}.

Let us denote,
\begin{align}
\beta_j = \sum_{k=1}^m \alpha_k U_{j,k},
\end{align}
We can then evaluate the Wigner function as, for the $n$-photon input,
\begin{align}
W(\vec\alpha) &= \frac{1}{\pi^{2m}}
\idotsint e^{\vec{\lambda}^*\cdot\vec{\alpha}-\vec{\lambda}\cdot\vec\alpha^*} e^{-\frac{1}{2}\mathcal{E}(\vec\lambda)} \nonumber \\
&\quad\times \prod_{j=1}^n \left(1-|\mu_j|^2\right) d^2\vec{\lambda} \nonumber \\
&=\frac{1}{\pi^{2m}}\idotsint e^{\vec{\mu}^*\cdot\vec{\beta}-\vec{\mu}\cdot\vec\beta^*} e^{-\frac{1}{2}\mathcal{E}(\vec{\mu})} \nonumber \\
&\quad\times \prod_{j=1}^n (1-|\mu_j|^2)  d^2\vec{\mu} \nonumber \\
&= \left(\frac{2}{\pi}\right)^m e^{-2|\vec{\beta}|^2} \prod_{j=1}^n (4|\beta_j|^2-1)  \nonumber \\
&= \left(\frac{2}{\pi}\right)^m e^{-2|\vec{\alpha}|^2} \prod_{j=1}^n \left(4\left|\sum_{k=1}^m \alpha_k U_{k,j}\right|^2-1\right). \nonumber \\
\end{align}
We have focused on the {\BS} case where the input state is \mbox{$\hat\rho=(\ket{1}\bra{1})^{\otimes n}\otimes (\ket{0}\bra{0})^{\otimes (m-n)}$}, and we will now consider a particular output probability $\mathcal{P}$, the one where a single photon is measured in the first $n$ output modes. This is determined by calculating the expectation value of the projector \mbox{$\hat\Pi = (\ket{1}\bra{1})^{\otimes n}\otimes (\ket{0}\bra{0})^{\otimes (m-n)}$}, which for this input state is equal to the expectation value of the $n$-dimensional number operator, \mbox{$\langle\hat{n}_1\ldots \hat{n}_n\rangle$}, where \mbox{$\hat{n}_j=\hat{a}_j^\dag\hat{a}_j$}. In the usual permanent-based approach, the absolute square of this amplitude corresponds to the permanent of a $n\times n$ submatrix of $U$.

Now, we will consider the phase-space approach for a particular output configuration of single photons at each output mode with $T=(1,\dots,1,0,\dots,0)$ and so $\mathcal{P}=|\mathrm{perm}(U^{n\times n})|^2$, where \mbox{$U^{n\times n}$} denotes the \mbox{$n\times n$} submatrix of $U$. For a single-mode state, the expectation value of the number operator is obtained from the Wigner function as \cite{bib:GerryKnight05},
\begin{align}
\langle \hat{n} \rangle =& \int W(\alpha) \left(|\alpha|^2-\frac{1}{2} \right) d^2\alpha,
\end{align}
where the \mbox{$|\alpha|^2-\frac{1}{2}$} term is obtained by expressing $\hat{n}$ in symmetrically ordered form \cite{bib:GerryKnight05}, \mbox{$\frac{1}{2}(\hat{a}^\dag\hat{a}+\hat{a}\hat{a}^\dag-1)$}, and making the well-known substitution \mbox{$\hat{a}^\dag \to \alpha^*$}, \mbox{$\hat{a}\to \alpha$}.

In the multimode case this expression generalizes to,
\begin{align} \label{eq:int_perm}
\mathcal{P} &= \langle \hat{n}_1 \ldots \hat{n}_n \rangle \nonumber \\
&= \idotsint W(\vec\alpha) \prod_{j=1}^n \left(|\alpha_j|^2-\frac{1}{2}\right) d^2\vec\alpha \nonumber \\
&= \left(\frac{2}{\pi}\right)^m \idotsint e^{-2|\vec{\alpha}|^2} \prod_{j=1}^n \left(4\left|\sum_{k=1}^m \alpha_kU_{k,j}\right|^2-1\right) \nonumber \\
&\quad\times \prod_{j=1}^n \left(|\alpha_j|^2 - \frac{1}{2}\right) d^2\vec{\alpha} \nonumber \\
&=  \left|\mathrm{perm}(U^{n\times n})\right|^2,
\end{align}
which is the primary result of this manuscript. That is, we have shown that these integrals are \mbox{\textbf{\#P}-hard} to evaluate, because they are equal to the square of the permanent. 

To simplify these integrals we use the identity,
\begin{align}
\label{idzer}
\int e^{-2|{\alpha}|^2}\left(|\alpha|^2 - \frac{1}{2}\right)d^2{\alpha} = 0.
\end{align}
This expression holds for each component of the vector $\vec{\alpha}$.
That means that we can simplify Eq. (\ref{eq:int_perm}) in the following way.
Expand out the first product in the third line to give a polynomial in $\alpha_k$ and $\alpha_k^*$. We can then eliminate terms in this polynomial by considering the integral when multiplied by \mbox{$\prod_{j=1}^n \left(|\alpha_j|^2 - \frac{1}{2}\right)$}.

First, if for any $k$, the terms contains $\alpha_k$ and $\alpha_k^*$ with a total power that is odd, then the integral is over an odd function, and must yield zero.
Second, if for any $k$ with \mbox{$1\le k \le n$}, the term does not contain $\alpha_k$ or $\alpha_k^*$, then the integral for that term must yield zero due to the identity in Eq. (\ref{idzer}).
Combining these two rules, we find that each remaining term must contain every $\alpha_k$, or $\alpha_k^*$ a nonzero (but even number) of times for each \mbox{$1\le k \le n$}.
Because the first product in Eq. (\ref{eq:int_perm}) is up to $n$, the maximum total power of $\alpha_k$ or $\alpha_k^*$ that can appear in any term is \mbox{$2n$}.
Since we must have each of these at least twice, they must appear exactly twice.

These considerations imply that all terms in the sum in the first product in Eq. (\ref{eq:int_perm}) for $m>n$ must yield zero.
Therefore the sum can be truncated to $n$ instead of $m$, which gives the expression in the form,
\begin{align} \label{eq:int_perm2}
\mathcal{P} &= \left(\frac{2}{\pi}\right)^m \idotsint e^{-2|\vec{\alpha}|^2} \left(4\left|\sum_{k=1}^n \alpha_kU_{k,j}\right|^2-1\right) \nonumber \\
&\quad\times \prod_{j=1}^n \left(|\alpha_j|^2 - \frac{1}{2}\right) d^2\vec{\alpha}.
\end{align}
In this expression, it is now clear that $\mathcal{P}$ can only depend on the \mbox{$n\times n$} submatrix of $U$.
That is to be expected because the probability should be given from the magnitude squared permanent of this submatrix.
In addition, the $-1$ term that appears in the third line in Eq. (\ref{eq:int_perm}) can only yield results that integrate to zero.
That means that we can further simplify $\mathcal{P}$ to,
\begin{align} \label{eq:int_perm3}
\mathcal{P} &= 4^n \left(\frac{2}{\pi}\right)^m \idotsint e^{-2|\vec{\alpha}|^2} \prod_{j=1}^n \left|\sum_{k=1}^n \alpha_kU_{k,j}\right|^2 \nonumber \\
&\quad\times \prod_{j=1}^n \left(|\alpha_j|^2 - \frac{1}{2}\right) d^2\vec{\alpha}.
\end{align}
Next, we can use the identity,
\begin{align}
\label{idpi}
\int e^{-2|{\alpha}|^2} d^2{\alpha} = \frac {\pi} 2,
\end{align}
and integrate over all $\alpha_k$ for \mbox{$k>n$} to give,
\begin{align} \label{eq:int_perm4}
\mathcal{P} &= \left(\frac{8}{\pi}\right)^n \idotsint e^{-2\sum_{j=1}^n|{\alpha_j}|^2} \prod_{j=1}^n \left|\sum_{k=1}^n \alpha_kU_{k,j}\right|^2 \nonumber \\
&\quad\times\prod_{j=1}^n \left(|\alpha_j|^2 - \frac{1}{2}\right) d^2\vec{\alpha}.
\end{align}To simplify further, we expand the first product of Eq.~\eqref{eq:int_perm4} into a sum of terms.
It is easy to check that,
\begin{align}
\label{idzer2}
\int e^{-2|{\alpha}|^2}\alpha^2\left(|\alpha|^2 - \frac{1}{2}\right)d^2{\alpha} = 0,
\end{align}
and similarly for $(\alpha^*)^2$.
That means that any terms where we have $\alpha_k^2$ or $(\alpha_k^*)^2$ will integrate to zero.
The only terms that do not integrate to zero are those with exactly the product \mbox{$|\alpha_1|^2|\alpha_2|^2\ldots |\alpha_n|^2$}.
It is also easy to check that,
\begin{align}
\label{idpi2}
\int e^{-2|{\alpha}|^2}|\alpha|^2\left(|\alpha|^2 - \frac{1}{2}\right)d^2{\alpha} = \frac{\pi} 8,
\end{align}
and therefore,
\begin{align}
\left(\frac{8}{\pi}\right)^n \idotsint e^{-2\sum_{j=1}^n|{\alpha_j}|^2} \prod_{j=1}^n |\alpha_j|^2 \nonumber \\
\quad\times \prod_{j=1}^n \left(|\alpha_j|^2 - \frac{1}{2}\right) d^2\alpha_1 \ldots d^2 \alpha_n = 1.
\end{align}
The value of $\mathcal{P}$ therefore corresponds to the coefficient of \mbox{$|\alpha_1|^2|\alpha_2|^2\ldots |\alpha_n|^2$} in the product \mbox{$\prod_{j=1}^n \left|\sum_{k=1}^n U_{k,j}\alpha_k\right|^2$}.
It is convenient to express this product as,
\begin{align}
\prod_{j=1}^n \left|\sum_{k=1}^n U_{k,j}\alpha_k\right|^2 = \left(\prod_{j=1}^n \sum_{k=1}^n U_{k,j}\alpha_k \right)\left(\prod_{j=1}^n \sum_{k=1}^n U_{k,j}^*\alpha_k^*\right).
\end{align}
To find the coefficient of \mbox{$|\alpha_1|^2|\alpha_2|^2\ldots |\alpha_n|^2$}, we can use the MacMahon Master theorem for permanents on both expressions in brackets on the right-hand side \cite{bib:macmahon60}.
The coefficient of \mbox{$\alpha_1\alpha_2\ldots\alpha_n$} in the expression in the first brackets is ${\rm perm}(U^{n\times n})$.
Similarly, the coefficient of \mbox{$\alpha_1^*\alpha_2^*\ldots\alpha_n^*$} in the expression in the second brackets is \mbox{${\rm perm}[(U^{n\times n})^*]$}.
Since the permanent of the complex conjugate is the complex conjugate of the permanent, this is equal to \mbox{$[{\rm perm}(U^{n\times n})]^*$}.
As a result, the coefficient of \mbox{$|\alpha_1|^2|\alpha_2|^2\ldots |\alpha_n|^2$} in the product \mbox{$\prod_{j=1}^n \left|\sum_{k=1}^n U_{k,j}\alpha_k\right|^2$} is \mbox{$|{\rm perm}(U^{n\times n})|^2$}.
That means that we can evaluate $\mathcal{P}$ as \mbox{$\mathcal{P} =|{\rm perm}(U^{n\times n})|^2$}.

Thus we find that we have a number of forms of the integral for the probability, and that this integral can be evaluated to the square of the permanent, which is exactly as we expect.  We could consider any of these intermediate forms of the integral, but will concentrate on Eq. (\ref{eq:int_perm}) for simplicity. Equation (\ref{eq:int_perm}) provides us with an equivalence between two alternate forms for a particular output amplitude of the linear optical network, the first expressed in terms of a matrix permanent, and the second in the form of a multidimensional integral. Because the former is known to be \mbox{\textbf{\#P}-hard} in the worst case, it follows that the latter is also.

While calculating matrix permanents is \mbox{\textbf{\#P}-hard} in the worst case, there are many special cases where symmetry or sparsity may be exploited to efficiently calculate the permanent. One such class of matrices is the permutation matrices \mbox{$\sigma\in S_m$}, elements of the symmetric group. We will show explicitly that the integral approach to calculating \mbox{$\mathrm{perm}(\sigma)$} is computationally efficient, consistent with existing understanding of the complexity of permanents.

When \mbox{$U=\sigma$}, we have the property,
\begin{align}
\sum_{k=1}^m  \alpha_kU_{j,k}^* = \alpha_{\sigma_j},
\end{align}
and we see that $\mathcal{P}$ 
is separable across $\vec{\alpha}$. In this case, the $n$-dimensional integral from Eq. (\ref{eq:int_perm}) also becomes separable, forming an $n$-dimensional product of integrals,
\begin{align} \label{eq:permutationResult}
\mathcal{P} &= \left(\frac{2}{\pi}\right)^m \idotsint e^{-2|\vec{\alpha}|^2} \prod_{k=1}^n \left(4\left| \alpha_{\sigma_k}\right|^2-1\right) \nonumber \\
&\quad\times \prod_{j=1}^n \left(|\alpha_j|^2 - \frac{1}{2}\right) d^2\vec{\alpha} \nonumber \\
&= \left(\frac{2}{\pi}\right)^m \Bigg\{ \prod_{j=1}^n \int e^{-2|\alpha_j|^2} \left(4\left| \alpha_{j}\right|^2-1\right) \nonumber \\
&\quad\times \left(|\alpha_j|^2 - \frac{1}{2}\right) d^2\alpha_j \Bigg\} \Bigg\{\prod_{j=n+1}^{m} \int e^{-2|\alpha_j|^2}  d^2\alpha_j \Bigg\}\nonumber \\
&= \left(\frac{2}{\pi}\right)^m \left(\frac{\pi}{2}\right)^n \left(\frac{\pi}{2}\right)^{m-n} \nonumber \\
&= 1,
\end{align}
where the integrals can easily be evaluated. Therefore the integrals in Eq.~\eqref{eq:permutationResult} are computationally efficient to evaluate given the separable structure, confirming that \mbox{$\mathrm{perm}(\sigma)=1$} for all $\sigma$ and $m$, as expected. (We note that the identity matrix is a member of the class $\sigma$.)

Therefore, as a special example, we have just shown that a particular class of matrices --- the permutation matrices --- are classically efficient to calculate in the integral formalism due to their separable structure. However, one can also work in the opposite direction, by finding integrals that are classically efficient to calculate, and then work backwards to argue that calculating the corresponding matrix permanent is computationally efficient. This illustrates the power of our quantum optics method for determining the computational complexity of certain multi-dimensional integrals.  
 
We have shown that by employing two alternate but equivalent approaches to expressing the output amplitudes of linear optics networks fed with single-photon inputs, we are able to provide a quantum optical equivalence between matrix permanents and a particular class of multidimensional integrals. This implies that this class of integrals is \mbox{\textbf{\#P}-hard} in the worst case. 

The equivalence provides two important insights with broad impact. Firstly, it demonstrates the \mbox{\textbf{\#P}-hardness} of these multi-dimensional integrals. Secondly, by expressing the permanent in integral form, existing knowledge of the structure of integrals provides further insight into the computational complexity of permanents. Finally, we have shown that tools from quantum optics can be used to prove results in computational complexity theory. 

\begin{acknowledgments}
We thank Ryan Mann, Michael Bremner, and Dave Wecker for helpful discussions. PPR acknowledges the financial support of Lockheed-Martin. DWB is funded by an ARC Future Fellowship (FT100100761) and an ARC Discovery Project (DP160102426). KRM acknowledges the Australian Research Council Centre of Excellence for Engineered Quantum Systems (Project number CE110001013). JPD acknowledges support from the Air Force Office (Grant No. FA9550-13-10098), the Army Research Office (Grant No. W911NF-13-1-0381), the National Science Foundation (Grant No. 1403105), and the Northrop Grumman Corporation.
\end{acknowledgments}


\clearpage
\appendix

\section{Evolution of displacement operators through a linear optics network} \label{app:disp_ev}

\newtheorem{lemma}{Lemma}
\begin{lemma}
A linear optical network maps a product of displacement operators over $m$ modes, with amplitudes $\lambda_j$, to another product of displacement operators with amplitudes given by \mbox{$\mu_j = \sum_{k=1}^m U_{j,k} \lambda_k$}.
\end{lemma}
\begin{proof}
\begin{align}
\hat{U} \left(\prod_{j=1}^m \hat{D}(\lambda_j)\right) \hat{U}^\dag &=\ \mathrm{exp}\left(\sum_{j=1}^m \hat{U}( \lambda_j \hat{a}^\dag_j - \lambda^*_j \hat{a}_j) \hat{U}^\dag\right) \nonumber \\
&=\ \mathrm{exp}\left(\sum_{j=1}^m \sum_{k=1}^m \lambda_j U_{j,k} \hat{a}^\dag_k - \lambda_j^* U_{j,k}^* \hat{a}_k\right) \nonumber \\
&=\ \prod_{j=1}^m \mathrm{exp}\left(\sum_{k=1}^m \lambda_j U_{j,k} \hat{a}^\dag_k - \lambda_j^* U_{j,k}^* \hat{a}_k\right) \nonumber \\
&=\ \prod_{j=1}^m \hat{D}(\mu_j),
\end{align}
where,
\begin{align}
\mu_j &= \sum_{k=1}^m \lambda_k U_{j,k},
\end{align}
which is Eq. (\ref{eq:beta_from_alpha}) that we set out to prove.
\end{proof}

\section{Overlap of the displacement operator with the single-photon state} \label{app:disp_ov}

\begin{lemma}
The overlap between the single-photon state with the displacement operator \mbox{$\bra{1} \hat{D}(\lambda) \ket{1}$} is given by \mbox{$e^{-\frac{1}{2}|\lambda |^2}(1-|\lambda |^2)$}.
\end{lemma}
\begin{proof}
\begin{align} \label{eq:charProof}
\bra{1}\hat{D}(\lambda)\ket{1} &=\ \bra{0}\hat{a} \hat{D}(\lambda) \hat{a}^\dag\ket{0} \nonumber \\
&=\ \bra{0}\hat{a} (\hat{a}^\dag -\lambda ^*) \hat{D}(\lambda)\ket{0} \nonumber \\
&=\ \bra{0}\hat{a}\hat{a}^\dag \hat{D}(\lambda) \ket{0} - \bra{0}\lambda^*\hat{a} \hat{D}(\lambda)\ket{0} \nonumber \\
&=\ \langle 0|\lambda\rangle - \lambda ^*\langle 1|\lambda\rangle \nonumber \\
&=\ e^{-\frac{1}{2}|\lambda|^2} - |\lambda |^2 e^{-\frac{1}{2}|\lambda|^2} \nonumber \\
&=\ e^{-\frac{1}{2}|\lambda|^2}(1 - |\lambda |^2),
\end{align}
which is the identity that we used to obtain Eq. (\ref{eq:INTmmCW}). We have used the commutation relation between the displacement operator and the photon-creation operator \cite{bib:GerryKnight05},
\begin{align}
[\hat{a}^\dagger, \hat{D}(\lambda)]=\lambda^*\hat{D}(\lambda),
\end{align}
as well as the coherent state represented in the Fock basis,
\begin{align}
\ket{\lambda}=e^{-\frac{|\lambda|^2}{2}}\sum_{n=0}^{\infty}\frac{\lambda^n}{\sqrt{n!}}\ket{n},
\end{align}
to read off the overlaps between the Fock states and the coherent state $\ket{\lambda}$.
\end{proof}


\begin{thebibliography}{1}
\expandafter\ifx\csname natexlab\endcsname\relax\def\natexlab#1{#1}\fi
\expandafter\ifx\csname bibnamefont\endcsname\relax
  \def\bibnamefont#1{#1}\fi
\expandafter\ifx\csname bibfnamefont\endcsname\relax
  \def\bibfnamefont#1{#1}\fi
\expandafter\ifx\csname citenamefont\endcsname\relax
  \def\citenamefont#1{#1}\fi
\expandafter\ifx\csname url\endcsname\relax
  \def\url#1{\texttt{#1}}\fi
\expandafter\ifx\csname urlprefix\endcsname\relax\def\urlprefix{URL }\fi
\providecommand{\bibinfo}[2]{#2}
\providecommand{\eprint}[2][]{\url{#2}}

\bibitem[{\citenamefont{Aaronson and Arkhipov}(2011)}]{bib:AaronsonArkhipov10}
\bibinfo{author}{\bibfnamefont{S.}~\bibnamefont{Aaronson}} \bibnamefont{and}
  \bibinfo{author}{\bibfnamefont{A.}~\bibnamefont{Arkhipov}},
  \bibinfo{journal}{Proc. ACM STOC (New York)} p. \bibinfo{pages}{333}
  (\bibinfo{year}{2011}).

\bibitem[{\citenamefont{Broome et~al.}(2013)\citenamefont{Broome, Fedrizzi,
  Rahimi-Keshari, Dove, Aaronson, Ralph, and White}}]{bib:Broome20122012}
\bibinfo{author}{\bibfnamefont{M.~A.} \bibnamefont{Broome}},
  \bibinfo{author}{\bibfnamefont{A.}~\bibnamefont{Fedrizzi}},
  \bibinfo{author}{\bibfnamefont{S.}~\bibnamefont{Rahimi-Keshari}},
  \emph{et al.}
  \bibinfo{journal}{Science} \textbf{\bibinfo{volume}{339}},
  \bibinfo{pages}{6121} (\bibinfo{year}{2013}).

\bibitem[{\citenamefont{Crespi et~al.}(2013)\citenamefont{Crespi, Osellame,
  Ramponi, Brod, Galvao, Spagnolo, Vitelli, Maiorino, Mataloni, and
  Sciarrino}}]{bib:Crespi3}
\bibinfo{author}{\bibfnamefont{A.}~\bibnamefont{Crespi}},
  \bibinfo{author}{\bibfnamefont{R.}~\bibnamefont{Osellame}},
  \bibinfo{author}{\bibfnamefont{R.}~\bibnamefont{Ramponi}},
  \emph{et al.}
  \bibinfo{journal}{Nat. Phot.} \textbf{\bibinfo{volume}{7}},
  \bibinfo{pages}{545} (\bibinfo{year}{2013}).

\bibitem[{\citenamefont{Tillmann et~al.}(2013)\citenamefont{Tillmann, Daki,
  Heilmann, Nolte, Szameit, and Walther}}]{bib:Tillmann4}
\bibinfo{author}{\bibfnamefont{M.}~\bibnamefont{Tillmann}},
  \bibinfo{author}{\bibfnamefont{B.}~\bibnamefont{Daki}},
  \bibinfo{author}{\bibfnamefont{R.}~\bibnamefont{Heilmann}},
  \emph{et al.}
  \bibinfo{journal}{Nat. Phot.} \textbf{\bibinfo{volume}{7}},
  \bibinfo{pages}{540} (\bibinfo{year}{2013}).

\bibitem[{\citenamefont{Spring et~al.}(2013)\citenamefont{Spring, Metcalf,
  Humphreys, Kolthammer, Jin, Barbieri, Datta, Thomas-Peter, Langford, Kundys
  et~al.}}]{bib:Spring2}
\bibinfo{author}{\bibfnamefont{J.~B.} \bibnamefont{Spring}},
  \bibinfo{author}{\bibfnamefont{B.~J.} \bibnamefont{Metcalf}},
  \bibinfo{author}{\bibfnamefont{P.~C.} \bibnamefont{Humphreys}},
  \emph{et al.}
  \bibnamefont{et~al.}, \bibinfo{journal}{Science}
  \textbf{\bibinfo{volume}{339}}, \bibinfo{pages}{6121} (\bibinfo{year}{2013}).

\bibitem[{\citenamefont{He et~al.}(2016)\citenamefont{He, Su, Huang, Ding, Qin,
  Wang, Unsleber, Chen, Wang, He et~al.}}]{bib:he16}
\bibinfo{author}{\bibfnamefont{Y.}~\bibnamefont{He}},
  \bibinfo{author}{\bibfnamefont{Z.-E.} \bibnamefont{Su}},
  \bibinfo{author}{\bibfnamefont{H.-L.} \bibnamefont{Huang}},
  \emph{et al.}
  \bibnamefont{et~al.}, \bibinfo{journal}{arXiv:1603.04127}
  (\bibinfo{year}{2016}).

\bibitem[{\citenamefont{Scheel and Buhmann}(2008)}]{bib:Scheel04}
\bibinfo{author}{\bibfnamefont{S.}~\bibnamefont{Scheel}} \bibnamefont{and}
  \bibinfo{author}{\bibfnamefont{S.~Y.} \bibnamefont{Buhmann}},
  \bibinfo{journal}{Acta Phys. Slovaca} \textbf{\bibinfo{volume}{58}},
  \bibinfo{pages}{675} (\bibinfo{year}{2008}).

\bibitem[{\citenamefont{Ryser}(1963)}]{bib:Ryser63}
\bibinfo{author}{\bibfnamefont{H.~J.} \bibnamefont{Ryser}},
  \bibinfo{journal}{Comb. Math., Carus Math. Mono.}
  \textbf{\bibinfo{volume}{14}} (\bibinfo{year}{1963}).

\bibitem[{\citenamefont{Gard et~al.}(2015)\citenamefont{Gard, Motes, Olson,
  Rohde, and Dowling}}]{bib:GardBSintro}
\bibinfo{author}{\bibfnamefont{B.~T.} \bibnamefont{Gard}},
  \bibinfo{author}{\bibfnamefont{K.~R.} \bibnamefont{Motes}},
  \bibinfo{author}{\bibfnamefont{J.~P.} \bibnamefont{Olson}},
  \emph{et al.}
  \bibinfo{chapter}{Chapter 8}, pp. \bibinfo{pages}{167--192}, in
  \emph{\bibinfo{title}{From Atomic to Mesoscale: The Role of Quantum Coherence in Systems of Various Complexities.}} \bibinfo{publisher}{World Scientific
  Publishing Co.} \bibinfo{year}{(2015)}, Eds. 
  \bibinfo{author}{\bibfnamefont{S.~A.} \bibnamefont{Malinovskaya}},
  \bibinfo{author}{\bibfnamefont{I.} \bibnamefont{Novikova}}.

\bibitem[{\citenamefont{Olson et~al.}(2015)\citenamefont{Olson, Seshadreesan,
  Motes, Rohde, and Dowling}}]{bib:olson2014sampling}
\bibinfo{author}{\bibfnamefont{J.~P.} \bibnamefont{Olson}},
  \bibinfo{author}{\bibfnamefont{K.~P.} \bibnamefont{Seshadreesan}},
  \bibinfo{author}{\bibfnamefont{K.~R.} \bibnamefont{Motes}},
  \emph{et al.}
  \bibinfo{journal}{Phys. Rev. A} \textbf{\bibinfo{volume}{91}},
  \bibinfo{pages}{022317} (\bibinfo{year}{2015}).

\bibitem[{\citenamefont{Seshadreesan et~al.}(2015)\citenamefont{Seshadreesan,
  Olson, Motes, Rohde, and Dowling}}]{bib:olson2014boson}
\bibinfo{author}{\bibfnamefont{K.~P.} \bibnamefont{Seshadreesan}},
  \bibinfo{author}{\bibfnamefont{J.~P.} \bibnamefont{Olson}},
  \bibinfo{author}{\bibfnamefont{K.~R.} \bibnamefont{Motes}},
  \emph{et al.}
  \bibinfo{journal}{Phys. Rev. A} \textbf{\bibinfo{volume}{91}},
  \bibinfo{pages}{022334} (\bibinfo{year}{2015}).

\bibitem[{\citenamefont{Rohde et~al.}(2015)\citenamefont{Rohde, Motes, Knott,
  Fitzsimons, Munro, and Dowling}}]{bib:catSampling}
\bibinfo{author}{\bibfnamefont{P.~P.} \bibnamefont{Rohde}},
  \bibinfo{author}{\bibfnamefont{K.~R.} \bibnamefont{Motes}},
  \bibinfo{author}{\bibfnamefont{P.~A.} \bibnamefont{Knott}},
  \emph{et al.}
  \bibinfo{journal}{Phys. Rev. A} \textbf{\bibinfo{volume}{91}},
  \bibinfo{pages}{012342} (\bibinfo{year}{2015}).

\bibitem[{\citenamefont{Lund et~al.}(2013)\citenamefont{Lund, Laing,
  Rahimi-Keshari, Rudolph, O'Brien, and Ralph}}]{bib:LundScatter13}
\bibinfo{author}{\bibfnamefont{A.~P.} \bibnamefont{Lund}},
  \bibinfo{author}{\bibfnamefont{A.}~\bibnamefont{Laing}},
  \bibinfo{author}{\bibfnamefont{S.}~\bibnamefont{Rahimi-Keshari}},
\emph{et al.}
  \bibnamefont{and} \bibinfo{author}{\bibfnamefont{T.~C.} \bibnamefont{Ralph}},
  \bibinfo{journal}{Phys. Rev. Lett.} \textbf{\bibinfo{volume}{113}},
  \bibinfo{pages}{100502} (\bibinfo{year}{2013}).

\bibitem[{\citenamefont{Rahimi-Keshari
  et~al.}(2015)\citenamefont{Rahimi-Keshari, Lund, and
  Ralph}}]{bib:SalehLOcomp}
\bibinfo{author}{\bibfnamefont{S.}~\bibnamefont{Rahimi-Keshari}},
  \bibinfo{author}{\bibfnamefont{A.~P.} \bibnamefont{Lund}}, \bibnamefont{and}
  \bibinfo{author}{\bibfnamefont{T.~C.} \bibnamefont{Ralph}},
  \bibinfo{journal}{Phys. Rev. Lett.} \textbf{\bibinfo{volume}{114}},
  \bibinfo{pages}{060501} (\bibinfo{year}{2015}).

\bibitem[{\citenamefont{Rahimi-Keshari
  et~al.}(2016)\citenamefont{Rahimi-Keshari, Ralph, and
  Caves}}]{bib:rahimi2015efficient}
\bibinfo{author}{\bibfnamefont{S.}~\bibnamefont{Rahimi-Keshari}},
  \bibinfo{author}{\bibfnamefont{T.~C.} \bibnamefont{Ralph}}, \bibnamefont{and}
  \bibinfo{author}{\bibfnamefont{C.~M.} \bibnamefont{Caves}},
  \bibinfo{journal}{Phys. Rev. X} \textbf{\bibinfo{volume}{6}},
  \bibinfo{pages}{021039} (\bibinfo{year}{2016}).

\bibitem[{\citenamefont{Motes et~al.}(2015)\citenamefont{Motes, Olson, Rabeaux,
  Dowling, Olson, and Rohde}}]{bib:MORDOR}
\bibinfo{author}{\bibfnamefont{K.~R.} \bibnamefont{Motes}},
  \bibinfo{author}{\bibfnamefont{J.~P.} \bibnamefont{Olson}},
  \bibinfo{author}{\bibfnamefont{E.~J.} \bibnamefont{Rabeaux}},
  \emph{et al.}
  \bibinfo{journal}{Phys. Rev. Lett.} \textbf{\bibinfo{volume}{114}},
  \bibinfo{pages}{170802} (\bibinfo{year}{2015}).

\bibitem[{\citenamefont{Huh et~al.}(2015)\citenamefont{Huh, Guerreschi,
  Peropadre, McClean, and Aspuru-Guzik}}]{bib:vibSpec}
\bibinfo{author}{\bibfnamefont{J.}~\bibnamefont{Huh}},
  \bibinfo{author}{\bibfnamefont{G.~G.} \bibnamefont{Guerreschi}},
  \bibinfo{author}{\bibfnamefont{B.}~\bibnamefont{Peropadre}},
  \emph{et al.}
  \bibinfo{journal}{Nat. Phot.} \textbf{\bibinfo{volume}{9}},
  \bibinfo{pages}{615} (\bibinfo{year}{2015}).

\bibitem[{\citenamefont{Reck et~al.}(1994)\citenamefont{Reck, Zeilinger,
  Bernstein, and Bertani}}]{bib:Reck94}
\bibinfo{author}{\bibfnamefont{M.}~\bibnamefont{Reck}},
  \bibinfo{author}{\bibfnamefont{A.}~\bibnamefont{Zeilinger}},
  \bibinfo{author}{\bibfnamefont{H.~J.} \bibnamefont{Bernstein}},
  \emph{et al.}
  \bibinfo{journal}{Phys. Rev. Lett.} \textbf{\bibinfo{volume}{73}},
  \bibinfo{pages}{58} (\bibinfo{year}{1994}).

\bibitem[{\citenamefont{Motes et~al.}(2014)\citenamefont{Motes, Gilchrist,
  Dowling, and Rohde}}]{bib:MotesLoop}
\bibinfo{author}{\bibfnamefont{K.~R.} \bibnamefont{Motes}},
  \bibinfo{author}{\bibfnamefont{A.}~\bibnamefont{Gilchrist}},
  \bibinfo{author}{\bibfnamefont{J.~P.} \bibnamefont{Dowling}},
  \emph{et al.}
  \bibinfo{journal}{Phys. Rev. Lett.} \textbf{\bibinfo{volume}{113}},
  \bibinfo{pages}{120501} (\bibinfo{year}{2014}).

\bibitem[{\citenamefont{Berry et~al.}(2004)\citenamefont{Berry, Scheel, Myers,
  Sanders, Knight, and Laflamme}}]{bib:berry2004post}
\bibinfo{author}{\bibfnamefont{D.~W.} \bibnamefont{Berry}},
  \bibinfo{author}{\bibfnamefont{S.}~\bibnamefont{Scheel}},
  \bibinfo{author}{\bibfnamefont{C.~R.} \bibnamefont{Myers}},
  \emph{et al.}
  \bibnamefont{and} \bibinfo{author}{\bibfnamefont{R.}~\bibnamefont{Laflamme}},
  \bibinfo{journal}{N. J. Phys.} \textbf{\bibinfo{volume}{6}},
  \bibinfo{pages}{93} (\bibinfo{year}{2004}).

\bibitem[{\citenamefont{Gerry and Knight}(2005)}]{bib:GerryKnight05}
\bibinfo{author}{\bibfnamefont{C.~C.} \bibnamefont{Gerry}} \bibnamefont{and}
  \bibinfo{author}{\bibfnamefont{P.~L.} \bibnamefont{Knight}},
  \emph{\bibinfo{title}{Introductory Quantum Optics}}
  (\bibinfo{publisher}{Cambridge University Press}, \bibinfo{year}{2005}).

\bibitem[{\citenamefont{MacMahon}(1960)}]{bib:macmahon60}
\bibinfo{author}{\bibfnamefont{P.}~\bibnamefont{MacMahon}},
  \emph{\bibinfo{title}{Combinatory Analysis, 1915--16}}
  (\bibinfo{publisher}{Chelsea Pub. Co.}, \bibinfo{year}{1960}).

\end{thebibliography}
\end{document}